\documentclass[11pt]{article}
\usepackage{fullpage,graphicx,algorithmic,algorithm,enumerate,epsfig,wrapfig}
\usepackage{color,amssymb,amsthm,thmtools}
\usepackage{hyperref}
\usepackage{amsmath}
\usepackage{geometry}
\declaretheorem[parent=section]{theorem}

\declaretheorem[sibling=theorem]{definition}

\declaretheorem[sibling=theorem]{corollary}

\declaretheorem[sibling=theorem]{lemma}

\DeclareMathOperator {\E}{\mathbb{E}}
\let\epsilon=\varepsilon

\newcommand{\poly}{\mathrm{poly}}
\newcommand{\sig}{\mathrm{sig}}

\newcommand{\q}{\widehat{d}}

\title{Reconciling Graphs and Sets of Sets}

\author{Michael Mitzenmacher\thanks{Harvard University School of Engineering and Applied Sciences.  email: michaelm@eecs.harvard.edu.  Michael Mitzenmacher was supported in part by NSF grants CNS-1228598, CCF-1320231, CCF-1563710 and CCF-1535795.} \and Tom Morgan\thanks{Harvard University School of Engineering and Applied Sciences.  email: tdmorgan@seas.harvard.edu.  Tom Morgan was supported in part by NSF grants CNS-1228598 and CCF-1320231.}}

\date{}

\begin{document}

\maketitle

\begin{abstract}
We explore a generalization of set reconciliation, where the goal is to
reconcile {\em sets of sets}.  Alice and Bob each have a parent set
consisting of $s$ child sets, each containing at most $h$ elements
from a universe of size $u$.  They want to reconcile their sets of
sets in a scenario where the total number of differences between all
of their child sets (under the minimum difference matching between
their child sets) is $d$.  We give several algorithms for this problem,
and discuss applications to reconciliation problems on graphs, databases,
and collections of documents.  We specifically focus on graph reconciliation, providing protocols based on set of sets reconciliation for random graphs from $G(n,p)$ and for forests of rooted trees.    
\end{abstract}

\section{Introduction}
In the standard problem of set reconciliation, two parties Alice and
Bob each hold sets of items $S_A$ and $S_B$ respectively from a common
universe (with items generally represented as words of $w$ bits), and the goal
is for one or both of the parties to determine the union of the two
sets.  Typically in applications the set difference size $d$ is small, and we
seek an algorithm that computes the set difference efficiently and with small communication.
When a bound $d$ on the size of the set difference is known, 
standard polynomial interpolation methods allow reconciliation using only $d$ words, but
this approach is fairly inefficient computationally \cite{minsky2003set,starobinski2003efficient}.  A more practical approach based on 
Invertible Bloom Lookup Tables (or IBLTs, also called Invertible Bloom Filters) uses 
$O(d)$ space and linear time, and succeeds with high probability \cite{eppstein2011straggler,eppstein2011s,gm11}.  

We explore an important problem variant of reconciling \emph{sets of sets}. 
Specifically, Alice and Bob each have a \emph{parent set} consisting of $s$ \emph{child sets}, each containing at most $h$ elements from a universe of size $u$.  We call the sum of the sizes of the child sets $n$.  Alice and Bob wish to reconcile their sets of sets under the scenario that the total number of differences among all of their child sets (under the minimum difference matching between their child sets) is $d$.  We consider primarily the one-way notion of reconciliation, in which at the end of the protocol, Bob can completely recover Alice's data.  (Our work can be extended to mutual reconciliation in various ways.)  

The problem of reconciling sets of sets naturally occurs in several reconciliation problems. For a theoretically interesting example that highlights the approach, we focus in this paper on using reconciliation of sets of sets for reconciling unlabeled random graphs.  Several graph isomorphism and graph watermarking techniques for random graphs develop what we call here \emph{signatures} for the vertices, where a signature corresponds to a set related to the neighborhood of the vertex \cite{babai1979canonical,czajka2008improved}.
We adapt this methodology for reconciliation,  reconciling a perturbed random graph by first reconciling the set of signatures of the vertices, which can themselves be represented as sets that have undergone a small number of total changes.  We also explore reconciling forest graphs, which are another class of graphs in which graph isomorphism is tractable \cite{aho1974design}. 

There are other settings where reconciling sets of sets is a natural primitive.  For example, consider relational databases consisting of binary data, where the columns are labeled but the rows are not.  In such a case, a row database entry can equivalently be thought of as a set of elements from the universe of columns (the set of columns in which the row has a 1 entry).  The problem of reconciling two databases in which a total of $d$ bits have been flipped corresponds exactly to our sets of sets problem.  As another example, consider the well-known approach of using shingles to represent documents \cite{broder1997resemblance}.  In this setting, consecutive blocks of $k$ words of a document are hashed into numbers, and a subset of these numbers are used as a signature for the document;  these signatures can be compared to determine document similarity.  A collection of documents would then correspond to sets of sets, and in cases where two collections had some documents that were similar (instead of exact matches), the corresponding sets would only have a small number of differences.  Reconciling collections of documents could start by reconciling the sets of sets corresponding to the collection, to find documents in one collection with no similar document in another collection.

In what follows, we first review related work in set reconciliation and other research areas.  We then review some known results regarding IBLTs before introducing a series of protocols for reconciling sets of sets.  This culminates in a one round protocol for reconciling sets of sets when $d$ is known in advance that uses $\tilde{O}(d)$ communication and $\tilde{O}(n+d^2)$ computation, as well as a 4 round protocol with matching bounds (up to log factors) when $d$ is unknown.  Finally, we examine applications, focusing on graph reconciliation problems.

\subsection{Related Work}
Initial work for set reconciliation considered two parties reconciling sets of numbers, and utilized characteristic polynomials \cite{minsky2003set,starobinski2003efficient}.  Given a set $S=\{x_1,x_2,\ldots,x_n\}$, Alice would compute the characteristic polynomial 
$\chi_S(z)=(z-x_1)(z-x_2)\cdots (z-x_n)$, and similarly Bob would compute the characteristic polynomial $\chi_T(z)$ for his set $T$.  If there are only $d$ differences in the two sets, then by sending any $d' > d$ evaluations of the characteristic polynomial at previously agreed upon points, Alice and Bob can recover each other's sets (by interpolating the rational function $\chi_S(z)/\chi_T(z)$).  Determining or estimating an upper bound on the number of differences $d$ is considered external to the reconciliation protocol (see \cite{eppstein2011s} for examples of set difference estimation, and Appendix A).  While this approach uses essentially minimal communication, the required interpolation is computationally expensive for even moderately large set differences.  

As an alternative, Invertible Bloom Lookup Tables (IBLTs) utilize an efficient peeling process.  Alice builds a hash table by hashing her set elements multiple times (generally 3 or 4, using multiple hash functions) into the table, and elements that collide are XOR'ed together.  Additional information, such as the number of elements that have hashed into a bucket, are also stored.  Alice sends the hash table to Bob, who then ``removes'' his elements from the table by XORing them in the same manner using the same hash functions into the table, so that only the set difference remains in the table.  Elements alone in a hash bucket can then be found and removed.  With $d$ differences, only a table of size $O(d)$ is needed to recover all the elements (with probability inverse polynomial in $d$), and only linear time is need to create and recover elements from the hash table.   

These ideas have been extended, for example to multi-party reconciliation \cite{boral2014multi,mitzenmacher2013simple}.  However, there remain many open problems related to reconciliation of objects beyond sets.  The problem of reconciling strings, or files, has received significant attention, with the program rsync being a well-known practical implementation \cite{rsync, rsyncalg}, and additional theoretical work \cite{irmak2005improved, yan2008algorithms}.  Reconciliation of objects such as graphs and databases that both limit the communication to be proportional to the size of the difference between the objects and are efficient computationally have received less attention.  For graphs, related work includes work on graph isomorphism \cite{babai2016graph, babai1979canonical,czajka2008improved}, which can be thought of as reconciling with no errors, and graph watermarking \cite{eppstein2016models}, which involves determining if one graph is a perturbation of another.  We incorporate approaches from these lines of work to construct reconciliation protocols for families of random graphs.

Graph reconciliation is also closely related to the problem of graph matching, also known as graph alignment or social network de-anonymization \cite{kazemi2015growing, kazemi2015can, korula2014efficient, yartseva2013performance}.  The underlying problem is essentially the same, but there are not two communicating agents;  one simply seeks an algorithm that aligns input graphs.  In the graph reconciliation setting, where the goal is to use small communication, the number of differences between the graphs is generally assumed to be very small, which is a setting largely unexplored in the graph matching setting.

Reconciling forests as a communication problem does not appear to have been studied previously.  Related work, again without communication, includes computing the edit distance between trees \cite{bille2005survey}.  Tree edit distance considers a different set of graph updates than we consider (vertex updates instead of edge updates). 

\section{Preliminaries} \label{sec:prelims}
Throughout this paper, we work in the word RAM model with words of size $w$.  We often refer to the number of rounds of communication a protocol uses, which denotes the number of total messages sent.  For example, a one round protocol consists only of a single message from Alice to Bob. 

All of our protocols assume access to public coins, meaning that any random bits used are shared between Alice and Bob at the cost of no additional communication.  This is relevant because our protocols lean heavily on the application of various hash functions, and public coins allow Alice and Bob to use the same hash functions without communicating anything about them. Note that our protocols can be converted to use only private coins with minimal additional communication via standard techniques \cite{newman1991private}.  In practice, one would generally start the protocol by sharing a small random seed to be used for generating all future randomness.

In set reconciliation, Alice and Bob each have sets $S_A$ and $S_B$ of size at most $n$ from a universe of size $u$ and their set difference, $|S_A \oplus S_B|$, is at most $d$.  We assume $w = \Omega(\log u + \log n)$.  In what follows we consider the one-way version of set reconciliation, in which at the end of the protocol Bob recovers Alice's set, as opposed to the two-way protocol, in which both parties recover the union of the two sets.

A tool we frequently use is the Invertible Bloom Lookup Table (IBLT) \cite{gm11}, a data structure that represents a set.  We briefly review its properties;  more details can be found in \cite{eppstein2011straggler,gm11}.  An IBLT is a hash table with $k$ hash functions and $m$ cells. We add a key to the table by updating each of the $k$ cells that it hashes to.  (We assume these cells are distinct;  for example, one can use a partitioned hash table, with each hash function having $m/k$ cells.)  Each cell maintains a count of the number of keys hashed to it, an XOR of all of the keys hashed to it, and an XOR of a checksum of all of the keys hashed to it.  The checksum, a sequence of $O(\log u)$ bits produced by another hash function, is sufficiently large so as to ensure that with high probability, none of the distinct keys' checksums collide.  We can also delete a key from an IBLT through the same operation as adding it, except that now we decrement the counts instead of incrementing them.

An IBLT is \emph{invertible} in that after adding $n$ unique keys to it, if $m$ is sufficiently large then we can recover those $n$ keys via a peeling process.  Whenever a cell in the table has a count of 1, its key value will be exactly the key hashed to that cell, which we can recover and then delete from the table.  This deletion potentially creates more cells with a count of 1, allowing the process to continue until no keys remain in the table.  To argue that this is likely to succeed, we interpret the IBLT as a random hypergraph with $m$ vertices and $n$ hyperedges of cardinality $k$.  Each cell corresponds to a vertex, and each key corresponds to a hyperedge connecting the $k$ cells it hashes to.  This peeling process then successfully extracts all of the hyperedges unless the graph has a nonempty 2-core, the probability of which can be directly bounded.  This gives the following theorem.

\begin{theorem}[Theorem 1 of \cite{gm11}] \label{thm:iblt}
There exists a constant $c$ so that an IBLT with $m$ cells ($O(m \log u)$ space) and at most $cm$ keys will successfully extract all keys with probability at least $1 - O(1/\poly(m))$.
\end{theorem}

Observe that if we allow the counts in an IBLT to become negative then we can ``delete'' keys that aren't actually in the table.  The IBLT can thus represent two disjoint sets, one for the added or ``positive'' keys and one for the deleted or ``negative'' keys.  The peeling process requires only a small modification to extract both sets.  When peeling, we now also peel cells with counts of $-1$ by adding the key back into the table.  As long as we only attempt to peel from cells that actually only have one key hashed there, both sets will be extracted.  Unfortunately, now a cell having a count of 1 or $-1$ might have multiple keys (some from each set) hashed there.  We use the checksum to check that a count of 1 or $-1$ corresponds to a single key.  By using enough bits for the checksum we can ensure that with suitably high probability, if the hash of the cell's key entry is equal to the cell's checksum, then there is only one key hashed to it.

Note that the IBLT extraction process now has two failure modes: peeling failures and checksum failures.  Peeling failures occur with at most $1 / \poly(m)$ and are entirely detectable as keys will remain in the IBLT that cannot be extracted.  Checksum failures occur with probability at most $1 / \poly(u)$ and may be undetectable, as an unlikely cascade of many of them in sequence can conceivably result in fully emptying the table but ending up with an incorrect set of keys.

The application of IBLTs to set reconciliation is immediate if Alice and Bob have an upper bound $d$ on the size of their set difference.  Alice constructs an $O(d)$ cell IBLT by adding each of her set elements to it.  She then sends it to Bob who deletes each of his set elements from it.  Bob then extracts all the keys from the IBLT.  Assuming the peeling succeeds, the extracted positive keys will be the set $S_A \setminus S_B$ and the extracted negative keys will be $S_B \setminus S_A$.  This gives us the following bound for set reconciliation.

\begin{corollary} \label{thm:setrecon}
Set reconciliation for known $d$ can be solved in 1 round using $O(d \log u)$ bits of communication and $O(n)$ time with probability at least $1-1/\poly(d)$.
\end{corollary}

We refer to ``encoding'' a set in an IBLT as inserting all of its elements into it.  We similarly ``decode'' a set difference from an IBLT by extracting its keys.  Also, in the above protocol it would be equivalent to have Bob delete all of his keys from an empty IBLT, and then combine Alice and Bob's IBLTs by adding/XORing the entries in each cell into an IBLT representing the difference between the sets.  We can in this way ``decode'' a set difference from a pair of IBLTs. 
We can ``recover'' a set by decoding a pair of IBLTs and applying the set difference to one of the sets the IBLTs were encoded from.  We often ward against checksum failures by augmenting the set recovery process with a hash of each of the sets.  Assuming the uniqueness of the hash, we can then detect checksum failures by comparing the recovered set to its hash.

As discussed in the previous section, there is another approach to set reconciliation which uses the characteristic polynomials of the two sets.  Although this is less efficient computationally than an IBLT, it succeeds with probability 1, which can be useful as a subroutine for set of sets reconciliation when working with very small set differences.

\begin{theorem}[\cite{minsky2003set}]
\label{thm:setreconpoly}
Set reconciliation for known $d$ can be solved in 1 round using $O(d \log u)$ bits of communication and $O(n \min(d, \log^2 n)+d^3)$ time with probability 1.
\end{theorem}

The running time comes from the time to compute the roots of the ratio of polynomials plus the time to evaluate the degree $n$ polynomials at $d$ points.  Computing the roots can be done in $O(d^3)$ via Gaussian elimination. (This can actually be done in time $O(d^{\omega+o(1)})$, where $\omega$ is the exponent for matrix inversion, although this approach is generally less efficient in practice.)  Evaluating the polynomials at $d$ points can be done one of two ways.  We can evaluate the polynomial in $O(n)$ time once for each of the points.  Alternatively, we can compute the coefficients of the polynomial in $O(n \log^2 n)$ steps via divide and conquer and polynomial multiplication using the Fast Fourier Transform.  Once we have the coefficients, we can evaluate $d \leq n$ roots of unity in $O(n \log n)$ time, again using the Fast Fourier Transform.

\section{Reconciling Sets of Sets}
\label{sec:sos}

In this section we discuss several protocols for reconciling sets of sets.  One tool we will frequently use is a \emph{set difference estimator}.  Such estimators were developed in \cite{eppstein2011s} for reconciliation problems;  their goal is to obtain a reasonably good estimate of the number of differences $d$, since our reconciliation algorithms generally require such a bound is known.  We improve on the estimators of \cite{eppstein2011s}, as we describe below.  

A set difference estimator is a data structure for estimating the size of the difference between two sets.  It implicitly maintains two sets $S_1$ and $S_2$ and supports three operations: update, merge, and query.  Update takes in an element $x$ and an index $i \in \{1,2\}$ and adds $x$ to $S_i$.  Merge takes in a second set difference estimator $D'$, which implicitly maintains sets $S'_1$ and $S'_2$ and returns a new set difference estimator $D''$ representing $S_1 \cup S'_1$ and $S_2 \cup S'_2$.  Query returns an estimate for $|S_1 \oplus S_2|$.

\begin{theorem} \label{thm:strata}
There is a set difference estimator requiring \\$O(\log(1/\delta) \log n)$ space with $O(\log(1/\delta))$ update, merge, and query times, which reports the size of the set difference to within a constant factor with probability at least $1 - \delta$.
\end{theorem}

This improves over the ``strata estimators'' of \cite{eppstein2011s} which take an additional $O(\log u)$ factor for space and an additional $O(\log n)$ factor for query and merge times.  Our estimators are built using techniques from streaming $\ell_0$-norm estimation and are discussed more in \autoref{app:strata}.

When engaging in set reconciliation without a known bound on $d$, Bob can send Alice a set difference estimator with all of his elements added to $S_1$.  Alice then creates a set difference estimator with her own set's elements in $S_2$, and then merges the two estimators and queries the merged estimator.  She then uses the estimate as a bound on $d$ for the protocol of \autoref{thm:setrecon}.  This gives us the following corollary.

\begin{corollary} \label{thm:setrecon2}
Set reconciliation for unknown $d$ can be solved in 2 rounds using $O(d \log u)$ bits of communication and $O(n \log d)$ time, with probability at least $1-1/\poly(d)$.
\end{corollary}

\subsection{Na\"ive Protocol}
In the problem of reconciling sets of sets, Alice and Bob each have a parent set of at most $s$ child sets, each containing at most $h$ elements from a universe of size $u$.  The sum of the sizes of the child sets is at most $n$.  Alice's set of sets is equal to Bob's after a series of at most $d$ element additions and deletions to Bob's child sets.  Let $\q$ be an upper bound on the number of child sets that differ between Alice and Bob.  In general, we may not have such a bound in which case we use $\q = \min(d,s)$.  We wish to develop protocols at the end of which, Bob can fully recover Alice's set of sets.

An equivalent way of defining $d$ is as the value of the minimum cost matching between Alice and Bob's child sets, where the cost of matching two sets is equal to their set difference.  Our protocols actually solve a slightly stronger problem than this.  All of our bounds hold for the setting where $d$ is the sum over each of Alice and Bob's child sets of their minimum set difference with one of the other party's child sets.  In other words, each child set needs to be mapped to at least one of the other party's child sets, but it doesn't have to be mapped to exactly one.

We assume that the word size is $w = \Omega(\log u + \log n)$.  We explore two versions of the problem, one in which the value $d$ is known (or a good upper bound on it is) which we call SSRK (Set of Sets Reconciliation for Known $d$) and the second version where $d$ is unknown, which we call SSRU (Set of Sets Reconciliation for Unknown $d$).

The simplest approach to reconciling sets of sets is to ignore the fact that the items are sets.  Instead we can just treat each child set as an item from a universe of size 
$$\sum_{i=0}^h {u \choose i} = O(\min(u^h, 2^u))$$ 
and directly reduce to set reconciliation (\autoref{thm:setrecon} and \autoref{thm:setrecon2}).  This approach gives the following results.

\begin{theorem} \label{thm:ssrnaive}
SSRK can be solved in one round using $$O(\q \min(h \log u, u))$$ bits of communication and $O(n)$ time with probability at least $1 - 1/\poly(\q)$.
\end{theorem}

\begin{theorem} \label{cor:ssrunaive}
SSRU can be solved in two rounds using $$O(\q \min(h \log u, u))$$ bits of communication and $O(n \log \q)$ time with probability at least $1-1/\poly(\q)$.
\end{theorem}

\subsection{IBLTs of IBLTs}

To go beyond the na\"{i}ve approach, we use a more compact representation of the child sets.  Specifically, we encode each child set in an IBLT (a \emph{child IBLT}).  As each of Alice's child sets differs from one of Bob's child sets by at most $d$ elements, it is possible to recover all of the child sets with only $O(d)$ cells per child IBLT, instead of the $h$ elements per child set we needed with the na{\"{i}ve solution.  We also include in our encoding a hash of the child set, as it allows us to easily identify which child set corresponds to which encoding.

\begin{algorithm}[h]
\caption{IBLT of IBLTs Protocol}
\label{alg:ibltofiblts}
\begin{itemize}
\item Alice encodes each of her child sets as an $O(d)$ cell IBLT.  She also uses a $O(\log s)$-bit pairwise independent hash function to compute a hash of each of her child sets.  We call the (child IBLT, hash) pair a child set's \emph{encoding}.  Let $E_A$ be her set of child encodings.
\item Alice creates an $O(\q)$ cell IBLT $T$, and $E_A$ into it.  Alice sends $T$ to Bob.
\item Bob computes the set of his child encodings $E_B$ and deletes it from $T$.  He decodes $T$ to find $E_A \setminus E_B$ and $E_B \setminus E_A$.
\item Bob computes $D_B$, his set of child sets whose hashes match one of the hashes in $E_B \setminus E_A$.
\item For each child IBLT $T_A \in E_A \setminus E_B$ Bob tries to decode it with each child IBLT $T_B \in E_B \setminus E_A$.  If there is no child IBLT $T_B$ with which it decodes, report failure.  Otherwise, Bob recovers Alice's child set corresponding to $T_A$ by applying the decoded set difference to his own child set from $D_B$ corresponding to $T_B$.  Bob creates $D_A$, the set of Alice's child sets that he recovers in this way.
\item Bob removes $D_B$ from his set of sets, and adds $D_A$.
\end{itemize}
\end{algorithm}

\begin{theorem} \label{thm:basic_compacting}
Algorithm \ref{alg:ibltofiblts} solves SSRK in one round using $$O( \q d \log u + \q \log s)$$ bits of communication and $O(n + \q^2 d)$ time with probability at least $1 - 1/\poly(\q)$.
\end{theorem}

\begin{proof}  
The only communication is through the transmission of $T$.  Each child encoding can be represented in $O(d\log u + \log s)$ space and $|E_A \oplus E_B| \leq \q$, so the communication cost is $O(\q d\log u + \q \log s)$, and $T$ decodes with probability at least $1 - 1 / \poly(\q)$ by \autoref{thm:iblt}.

The running time for computing the encodings, inserting them into/deleting them from $T$, and  decoding $T$ is $O(n)$.  The time for Bob to recover all of Alice's differing child sets is $O(\q^2 d)$, since there are $|E_A \setminus E_B| \cdot |E_B \setminus E_A| = O(\q^2)$, pairs of child IBLTs to attempt to decode, and each decoding takes $O(d)$ time.  Bob's remaining operations can then all be done in $O(n)$ time, giving us a total time of $O(n + \q^2 d)$.

This protocol succeeds so long as none of the child hashes collide, $T$ decodes, and each child IBLT in $E_A \setminus E_B$ successfully decodes with at least one child IBLT in $E_B \setminus E_A$.  There are $O(s^2)$ pairs of child hashes, and each pair collides with probability $1 / \poly(s)$, so by choosing the constant in our $O(\log s)$-bit hash large enough, none of them collide with probability at least $1 - 1 / \poly(s)$.  $T$ fails to decode with probability at most $1 / \poly(d)$.  Each of Alice's differing child sets is within $d$ elements of one of Bob's differing child sets, so for each child IBLT in $E_A \setminus E_B$, there is at least one child IBLT in $E_B \setminus E_A$ with which it will decode with probability at least $1 - 1/\poly(d)$.  As $|E_A \setminus E_B| \leq \q$, all child IBLTs in $E_A \setminus E_B$ decode with probability at least $1 - 1/\poly(d)$.  Therefore, the protocol fails with probability at most
$$1 / \poly(s) + 1 / \poly(\q) + 1/\poly(d) = 1/\poly(\q).$$
\end{proof}

Here we have focused on reconciling sets of sets with a small number of total differences.  However, we note that this approach can also be used to reconcile sets of sets where some child sets may not have a close match.  For example, in our proposed application for reconciling collections of documents, we would expect most documents to be exact duplicates, some to be near-duplicates, and some to be ``fresh'', non-duplicate documents.  We could use the approach of Theorem~\ref{thm:basic_compacting} to find near-duplicate and non-duplicate documents;  non-duplicate documents would yield child IBLTs that could not be decoded when compared against other child IBLTs, which could then be later transmitted directly.

When $d$ is unknown, we cannot simply use a single set difference estimator as we need to know both a bound on the number of differing child sets, and the maximum number of differences between a child set and its closest match on the other side.  We instead handle this scenario through the standard repeated doubling trick of trying the protocol for $d = 1,2,4,8,\ldots$ until the protocol succeeds.  This does not affect the asymptotic communication cost, but it does require $O(\log d)$ rounds of communication and gives us the following result.

\begin{corollary} \label{cor:basic_compacting}
SSRU can be solved in $O(\log d)$ rounds using $$O(\q d \log u + \q \log s)$$ bits of communication and $O(n \log d + \q^2 d)$ time with probability at least $1 - 1/\poly(\q)$.
\end{corollary}

So far we have used that there are at most $d$ differing child sets and each child set differs from another by at most $d$, but we have not exploited the fact that there are a total of $O(d)$ changes across all of the child sets, rather than $O(d^2)$.  We rectify this in the next result by observing that only $O(1)$ of the child sets need $\Omega(d)$ cells, $O(\sqrt{d})$ of the child IBLTs need $\Omega(\sqrt{d})$ cells, and so forth.

\begin{theorem} \label{thm:smarter_compacting}
Algorithm \ref{alg:ibltofiblts2} solves SSRK in one round using $$O(d\log \min(d,h) \log u + d\log s)$$ bits of communication and $O(n \log \min(d,h) + \q d\log \q)$ time with probability at least $2/3$.
\end{theorem}

\begin{proof}
Going forward we will condition on the event that there are no collisions among the $O(\log (st))$ bit hashes in the child encodings.  There are at most $2s$ child sets per round that can collide, so union bounding over the $t = \log_2 \min(d,h)$ rounds we have no collisions with probability at least $1-4s^2t/\poly(st)\geq 1-1/30$.  

Let us divide Alice's child sets into groups according to how many elements they differ by under the minimum difference matching.  $S_j$ is the set of Alice's child sets whose set difference with its match is in $[2^{j-1}, 2^j - 1]$.  First, observe that every one of Alice's differing child sets is included in some $S_j$ for $j \leq \log d + 1$.  Second, observe that $|S_j| \leq d / 2^{j-1}$ since the total number of element changes is at most $d$.

Consider the IBLT $T_i$. We choose the constant factor parameters of the child IBLTs such that if the child IBLT has fewer than $2^i$ items in it, it successfully decodes with probability at least $1 - \frac{1}{100 \cdot 2^i}$.  Let $Y_i$ be the event that IBLT $T_i$ successfully decodes.  Conditioned on $Y_i$, when processing to match up the child sets within $T_i$, in expectation Bob fails to recover at most $\frac{1}{100 \cdot 2^i}$ of Alice's child sets from $\cup_{j=1}^i S_j$ that he has not yet recovered.  By Markov's inequality, Bob recovers fewer than 9/10 of Alice's child sets in $\cup_{j=1}^i S_j$ that he has not already decoded with probability at most $\frac{1}{10 \cdot 2^i}$.  We use $X_i$ to refer to the event that processing $T_i$ results in Bob recovering at least $9/10$ of $\cup_{j=1}^i S_j$ that he had not previously recovered, so we have argued that
$$\Pr[X_i | Y_i] \geq 1 - \frac{1}{10 \cdot 2^i}.$$

Since there are at most $2d$ differing child sets in $T_1$, we can choose the IBLT's parameters so that $Y_1$ occurs with probability at least $1 - \frac{2}{10d}$.  For $i > 1$, conditioned on $\cap_{j=1}^{i-1}X_j$, the number of Alice's child sets left to be recovered after $T_i$ is processed is at most
\begin{align*}
\sum_{j=i+1}^t &|S_j| + \sum_{j=1}^i |S_j| 10^{j-i-1} \\
&\leq \sum_{j=i+1}^t d / 2^{j-1} + \sum_{j=1}^i d 10^{j-i-1} / 2^{j-1} \\
&\leq d / 2^{i-1} + d / 10^i \sum_{j=1}^i 5^{j-1} \\
&\leq d / 2^{i-1} + d / 2^{i+2} = (9 / 4) (d / 2^i).
\end{align*}
Since $T_i$ has $O(d / 2^i)$ cells, we can choose the constant factors in the order notation so that $Y_i$ occurs with probability at least $1 - \frac{2^i}{10 d}$ conditioned on $\cap_{j=1}^{i-1}X_j$. 

\begin{algorithm}[h]
\caption{Cascading IBLTs of IBLTs Protocol}
\label{alg:ibltofiblts2}
\begin{itemize}
\item For $i = 1, \ldots, t = \log_2 \min(d,h)$, Alice creates an ($O(2^i)$ cell child IBLT, $O(\log (st))$ bit hash) child encoding for each of her child sets and inserts it into an $O(d / 2^i)$ cell IBLT $T_i$.
\item If $t = \log_2 h$, Alice creates an $O(d/h)$ cell IBLT $T_*$ and inserts an $O(h \log u)$ bit encoding of each of her child sets into it.
\item Alice sends $T_1,\ldots,T_t$ and $T_*$ to Bob.
\item Bob deletes ($O(1)$ cell IBLT, hash) encodings of each of his child sets from $T_1$, and then extracts all of the different child encodings from it.  He uses the hashes of his extracted child sets to recover $D_B$, the set of his child sets that differ from any of Alice's.  
\item Bob tries every combination of the extracted child IBLTs, trying to recover Alice's child sets by finding a matching child IBLT of his own.  He inserts each of Alice's child sets that he recovers into the set $D_A$.  Going forward, he will recover more and more of Alice's child sets and $D_A$ will be the set he has recovered so far.
\item For each $i = 2,\ldots,t$, Bob performs the following procedure.  He first deletes the ($O(2^i)$ cell IBLT, hash) of each of his child sets from $T_i$, except for those in $D_B$. He also deletes the ($O(2^i)$ cell IBLT, hash) encoding of each child set in $D_A$ from $T_i$.  He then decodes $T_i$ and extracts all of the different child encodings, which correspond exactly to Alice's differing child sets that aren't yet in $D_A$.  He tries to decode each of Alice's extracted child IBLTs with the child IBLT of each set in $D_B$, adding Alice's child sets that he recovers to $D_A$.
\item If Bob received $T_*$, he deletes all of his child sets from it.  He also deletes each child set in $D_A$ from it.  He then decodes $T_*$ and adds all of the decoded child sets to $D_A$.
\item Bob deletes $D_B$ from his set of sets, and adds $D_A$.
\end{itemize}
\end{algorithm}

If $t < \log h$, and therefore there is no $T_*$, Bob successfully recovers all of Alice's child sets so long as all $X_i$ and $Y_i$ occur.  The probability of this is
\begin{align*}
\Pr&[\cap_{i=1}^t (X_i \cap Y_i)] \\
&= \Pr[Y_1] \Pr[X_1 | Y_1] \ldots \Pr\left[Y_t | \cap_{j=1}^{t-1}X_j\right] \Pr[X_t | Y_t] \\
&= \Pr[Y_1] \prod_{i=2}^t \Pr\left[Y_i | \cap_{j=1}^{i-1}X_j\right] \prod_{i=1}^t \Pr[X_i | Y_i]  \\
&= \prod_{i=1}^t \left(1 - \frac{2^i}{10d}\right) \prod_{i=1}^t \left(1 - \frac{1}{10 \cdot 2^i}\right) \\
&\geq 1 - \sum_{i=1}^t \left(\frac{2^i}{10d} + \frac{1}{10 \cdot 2^i}\right) \\
&\geq 4/5.
\end{align*}
If $t = \log h$, then the protocol will succeed so long as $T_*$ successfully decodes.  $T_*$ has $O(d/h)$ cells and if all $X_i$ and $Y_i$ occur then there are at most $(9/4)d/h$ elements to extract from $T_*$, so we can choose the constants such that $T_*$ decodes with probability at least $1/10$.  We have thus proved that by the end of the procedure, Bob recovers Alice's set of sets with probability at least $4/5-1/10-1/30=2/3$.

The time for Alice and Bob to insert or delete all of their child encodings from the $T_{i}$ is $O(n \log \min(d,h))$ since for each of the $t$ tables, it takes them $O(n)$ time to iterate through their child sets and construct all of their child encodings.  The remaining time is what Bob takes to attempt to decode Alice's child IBLTs.  When processing $T_i$, Bob extracts $O(\min(\q, d/2^i))$ of Alice's child IBLTs, and compare each one against each of his $O(\q)$ differing child sets' IBLTs.  Each child IBLT has $O(2^i)$ cells, so the total processing time is $O(\q \min(\q, d/2^i) 2^i)$.  Summing over $i$, we get $O(\q d \log \q)$.

The communication cost of transmitting $T_{1},\ldots,T_{t}$ is
\begin{align*}
O&\left(\sum_{i=1}^t (d/2^i) \cdot (\log (st) + 2^i \log u)\right) \\
&= O(d \log (st) + d t \log u)\\
&= O(d \log s + d \log \min(d,h) \log u).
\end{align*}
The cost of $T_*$ is $O(d/h \cdot h \log u) = O(d \log u)$.
\end{proof}

We can once again extend this protocol to SSRU by repeated doubling on $d$.

\begin{corollary} \label{cor:smarter_compacting}
SSRU can be solved in $O(\log d)$ rounds using $$O(d\log \min(d,h) \log u + d\log s)$$ bits of communication and $O(n \log d \log \min(d,h) + \q d\log \q)$ time with probability at least $2/3$.
\end{corollary}

We note that we could extend this recursive use of IBLTs further -- creating IBLTs of structures representing sets of sets as IBLTs of IBLTs -- to reconcile sets of sets of sets, but we do not currently have a compelling application, and leave further discussion to future work. 

All of our protocols so far have given specific failure probabilities, but any of them can be amplified to achieve an arbitrarily small failure probability via replication.  Specifically, Alice can send Bob a hash of her whole set of sets.  They can then run the protocol many times in parallel, with Bob outputting the first recovered set of sets to match Alice's hash.

\subsection{A Multi-Round Approach}
\label{sec:multi}
All of our previous protocols required only a single transmission when $d$ is known.  What follows is a protocol that exploits the additional power of a larger, but still constant, number of rounds of communication.  This protocol, through its use of set difference estimators, also extends more efficiently to the case where $d$ is unknown.

This protocol proceeds in four steps.  First, if $d$ is unknown, Alice and Bob estimate the number of different child sets they have by exchanging set difference estimators of the set of hashes of their child sets.  Second, Alice and Bob determine which of their child sets differ by performing set reconciliation on the hashes of their child sets.  Third, they exchange set difference estimators for the differing child sets, so that for each of Alice's differing child sets, they can identify one of Bob's differing child sets that is similar to it.  Finally, they use these matchings to engage in parallel set reconciliation of all of their differing child sets.  For child sets with larger differences they use \autoref{thm:setrecon}, and for smaller differences they use \autoref{thm:setreconpoly}.

\begin{restatable}{theorem}{multi}
\label{thm:strata_based}
SSRK can be solved in 3 rounds using $$O(\lceil\log_{\q}(1/\delta)\rceil \q\log s + \log (\q / \delta) \q\log h + \lceil\log_d(1/\delta)\rceil d \log u)$$ bits of communication and $$O(\log (\q / \delta) (n+\q^2)+d^2+\min(d h, n \sqrt{d}, n \log^2 h))$$ time with probability at least $1 - \delta$.
\end{restatable}

\begin{restatable}{theorem}{multiu} \label{thm:strata_based2}
SSRU can be solved in 4 rounds using $$O(\lceil\log_{\q}(1/\delta)\rceil \q\log s + \log (\q / \delta) \q\log h + \lceil\log_d(1/\delta)\rceil d \log u)$$ bits of communication and $$O(\log (\q / \delta) (n+\q^2)+d^2+\min(d h, n \sqrt{d}, n \log^2 h))$$ time with probability at least $1 - \delta$.
\end{restatable}

The proofs of these theorems appear in \autoref{app:multi}.

\subsection{Handling Multisets}
For some applications, it is useful to perform set reconciliation or set of sets reconciliation with multisets.  For multiset reconciliation, which is set reconciliation with multisets, \autoref{thm:setreconpoly} works as is.  \autoref{thm:setrecon} and \autoref{thm:setrecon2} can be made to work with a simple modification.  We create a set from our multiset, where if an element $x$ occurs in the multiset $k$ times, then $(x,k)$ is an element of the set.  After reconciling this set, recovering the corresponding multiset is immediate.  All of the bounds stay the same ($d$ can only decrease), except that $u$ grows to $u\cdot n$.  All of our protocols can be adapted to reconciling sets of multisets or multisets of multisets in a similar way.  

\subsection{Comparison of Results}
Our protocols have numerous parameters and can therefore be difficult to compare.  Depending on the parameters, and the relative importance of time, rounds of communication and total communication, any of them can be superior.  Here, for intuition, we compare them under an example application.

\autoref{tab:ssrkpoly} compares our protocols in a natural setting of parameters for reconciling relational databases of binary data.  Assuming the data is sufficiently dense in 1s, $h = \Theta(u)$ and $n = \Theta(su)$.  We specifically look at the case when $d$ is very small (smaller than $s$ and $h$) and we wish the reconciliation to succeed with high probability (the protocols are replicated until the failure probability $\delta = 1 / \poly(n)$).  For sufficiently large $u$, the protocols are sorted in ascending order of their total communication costs.  However, \autoref{thm:strata_based} takes more rounds than the rest, and Theorems \ref{thm:ssrnaive}, \ref{thm:basic_compacting} and \ref{thm:smarter_compacting} are in descending order of computation time (assuming sufficiently small $d$).

\begin{table*} 
	\centering
    \begin{tabular}{| l | l | l | l |}
    \hline
    Algorithm & Communication & Time & Rounds \\ \hline
    \autoref{thm:ssrnaive} & $d u \log n / \log d$ & $n \log n / \log d$ & 1 \\ \hline
    \autoref{thm:basic_compacting} & $d \log n (d\log u + \log s) / \log d$ & $(n + d^3) \log n / \log d$ & 1 \\ \hline
    \autoref{thm:smarter_compacting} & $d \log n (\log d \log u + \log s)$ & $(n + d^2) \log n \log d$ & 1 \\ \hline
    \autoref{thm:strata_based} & $d \log n (\log u + \log s / \log d)$ & $(n + d^2)\log n$ & 3 \\ \hline
    \hline
    \end{tabular}
	\caption{A comparison of results of the protocols for SSRK when $h = \Theta(u)$, $n = \Theta(su)$, $\delta = 1 / \poly(n)$ and $d \leq s,h$.  All of the communication and time bounds omit constant factors.}
	\label{tab:ssrkpoly}
\end{table*}

\section{General Graph Reconciliation}
We now consider applications to graph reconciliation.  Here Alice and Bob each have an unlabeled graph, $G_A = (V_A, E_A)$ and $G_B = (V_B, E_B)$ respectively, where $|V_A|=|V_B|=n$, and only $d << n$ edges need to be changed (added or deleted) in $E_A$ to make $G_A$ isomorphic to $G_B$.  Alice and Bob wish to communicate so that they both end up with the same graph.  We note that if $G_A$ and $G_B$ were labeled graphs, then the problem would be equivalent to set reconciliation on their sets of labeled edges.

There are various ways to formalize what we want the final graph to be.  The simplest is the \emph{one-way} version of the problem, where we want Bob to end up with a graph isomorphic to $G_A$.  Alternatively, we might wish for them to end up with something corresponding to the union of the two graphs, but this is not always well defined.  Figure \ref{fig:no_union} gives an example of two graphs where adding one edge to each will yield isomorphic graphs, but there are multiple distinct ways to do so which yield graphs that are not isomorphic.  Consequently, we focus on one-way reconciliation, and simply note that our techniques can be generally be extended to most natural two-way versions.

\begin{figure}
	\centering
 		\includegraphics[width=.45\textwidth]{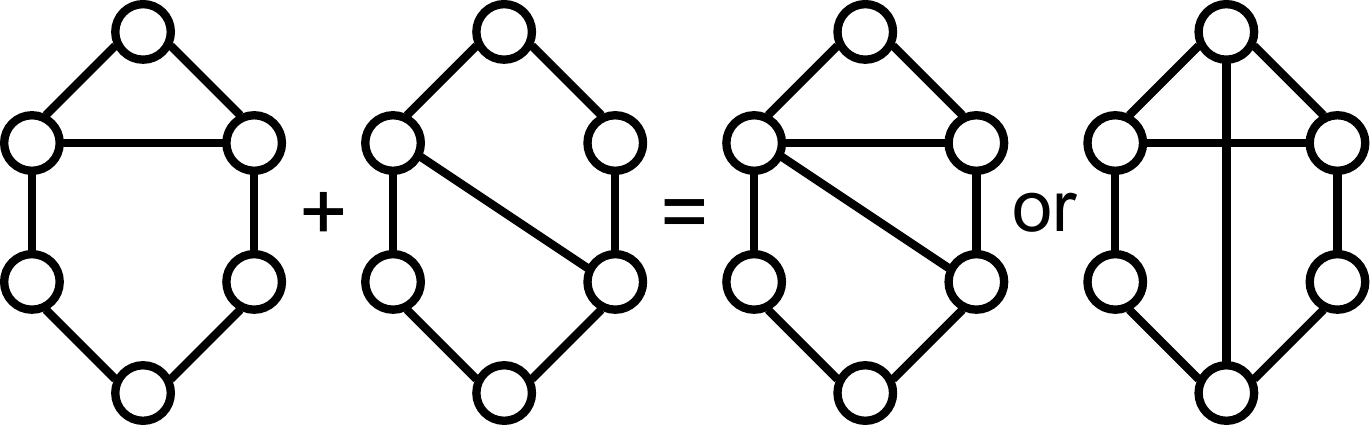}
	\caption{Depending on how we rotate these two graphs before merging them, we get different (non-isomorphic) results.  There is no way to add an edge to only one graph and get isomorphic results, but here we show two different ways to add one edge to each graph and get isomorphic results.}
\label{fig:no_union}
\end{figure}

Before exploring more computationally efficient protocols for graph reconciliation, we investigate what is possible when Alice and Bob each have access to unlimited computation.  This gives us bounds on what we can hope to achieve efficiently.  First let us look at graph isomorphism in this model.  Here we only wish to determine if $G_A$ is isomorphic to $G_B$.  The following protocol is apparently folklore.

\begin{itemize}
\item Alice iterates through all graphs in increasing lexicographical order until she finds one that is isomorphic to hers (checking by exhaustive search).  As there are at most $2^{n \choose 2}$ possible graphs, she can write down the index of the graph as a binary string $s_A$ of length ${n \choose 2}$.  Let $p_A(x)$ be a polynomial of degree ${n \choose 2}$ over $\mathbb{Z}_q$ for some suitably large prime $q$ with coefficients equal to the bits of $s_A$.  She picks a number $r$ uniformly at random from $\mathbb{Z}_q$, and sends Bob $r$ and $p_A(r)$.

\item Bob finds the index of the first graph in increasing lexicographical order which is isomorphic to $G_B$.  Similarly to Alice, he determines a corresponding polynomial $p_B(x)$, and then checks if $p_B(r) = p_A(r)$ in $\mathbb{Z}_q$.  If they are equal, he reports that $G_A$ is isomorphic to $G_B$, and otherwise he reports that they are not isomorphic.
\end{itemize}

\begin{theorem}
This protocol uses $O(\log q)$ bits of communication and succeeds with probability $1-O(n^2/q)$.
\end{theorem}

\begin{proof}
The only communication used is Alice sending $r$ and $p_A(r)$, and both of these can be represented with $O(\log q)$ bits.

If $G_A$ is isomorphic to $G_B$ then $s_A = s_B$ and $p_A = p_B$ so $p_A(r) = p_B(r)$ in $\mathbb{Z}_q$ and Bob will always report that the graphs are isomorphic.

If $G_A$ is not isomorphic to $G_B$ then $s_A \neq s_B$ and $p_A \neq p_B$.  $p_A(r)$ will equal $p_B(r)$ if and only if $r$ is a root of $p_A(x) - p_B(x)$, which a polynomial of degree $O(n^2)$.  By the Schwartz-Zippel Lemma \cite{schwartz1980fast, zippel1979probabilistic}, this occurs with probability at most $O(n^2 / q)$.
\end{proof}

\begin{corollary}
Graph isomorphism can be solved with $O(\log n)$ bits of communication with probability at least $1 - 1/n$.
\end{corollary}

Now we turn back to the problem of one-way graph reconciliation.  We assume that at most $d$ edges need to be added or deleted from $G_B$ to yield a graph isomorphic to $G_A$.

\begin{theorem} \label{thm:com_alg}
Graph reconciliation can be solved with probability at least $1 - 1/n$ using one round of $O(d \log n)$ bits of communication.
\end{theorem}

\begin{proof}
Alice's part of the protocol is the same as for graph isomorphism (for a specific choice of $q$ to be determined later).  She computes and transmits $r$ and $p_A(r)$ to Bob using a total of $O(\log q)$ bits of communication.  Bob then finds the polynomial corresponding to every possible change of $d$ edges in $G_B$.  For each of these polynomials, he evaluates them at $r$ and checks if the result equals $p_A(r)$.  Bob then changes his graph to the first graph $G_C$ he encounters where $p_C(r) = p_A(r)$.

There are $\sum_{i=0}^d \left({\binom{n}{2}}\atop{i}\right) = O(n^{2d})$ graphs that Bob will consider, and we are guaranteed that at least one of them is isomorphic to $G_A$.  The protocol fails only if one of the graphs that is not isomorphic to $G_A$ matches its polynomial value.  For a given graph, this will occur with probability at most $O(n^2 / q)$, thus we can union bound over all of the graphs to get a total failure probability of at most $O(n^{2d+2} / q)$.  We now choose $q = n^{2d+3}$ to get our desired probability.  The total communication is now $O(\log n^{2d+3}) = O(d \log n)$.
\end{proof}

This communication bound is also tight.  By a simple encoding argument, we derive the following.  
\begin{theorem}
Any protocol solving graph reconciliation with probability $1-O(1)$ must use $\Omega(d \log n)$ bits of communication in expectation.
\end{theorem}

\begin{proof}
Given a protocol $P$ solving reconciliation with probability $p$ and using $o(d \log n)$ bits of communication in expectation, we show how to encode $d \log n$ bits into $(1-p) d \log n + o(d \log n) < d \log n$ bits on average, yielding a contradiction.  

Let $s \in \{0,1\}^{d\log n}$, and $s = \langle s_1,\ldots,s_d \rangle$ be $s$ divided into strings of length $\log n$.  Let $G_B$ be a graph consisting of three disjoint sets of vertices $V_1$, $V_2$ and $V_3$, where $V_1 = \{v_1,\ldots,v_d\}$, $V_2 = \{v_{d+1},\ldots,v_{d+n}\}$ and $|V_3| = O(n+d)^2$.  There are no edges between any of the vertices in $V_1$ and $V_2$, and each vertex in $V_3$ has degree 1.  For each $i \in [d+n]$, $v_{i}$ is connected to $i+1$ vertices in $V_3$.  Overall there are $O(n+d)^2$ vertices, and each vertex in $V_1$ and $V_2$ is distinguished by its number of adjacent vertices of degree one.  

Let $G_A$ be identical to $G_B$, except with an additional $d$ edges.  For each $i \in [d]$, $v_i$ has an edge to $v_{d+s_i}$, where here we are interpreting $s_i$ as an integer in $[n]$.  Let $m$ be the transcript of communication between Alice and Bob that $P$ would create given Alice's graph is $G_A$ and Bob's is $G_B$.  Let $G'_A$ be the graph $P$ has Bob create given $m$ and $G_B$.  If $G'_A$ is not isomorphic to $G_A$, then $P$ failed and we use $\langle 0, s \rangle$ as our encoding.  If $G'_A$ is isomorphic to $G_A$, then our encoding is $\langle 1, m \rangle$.

Whenever $P$ succeeds, and the encoding is  $\langle 1, m \rangle$, we can recover $s$.  Since $G'_A$ is isomorphic to $G_A$, we can identify $V_1$ and $V_2$ unambiguously, and then determine the $s_i$ values via the edges connecting $V_1$ and $V_2$.  The average number of bits we encode $s$ into is 
\begin{align*}
(1-p) &d \log n + p \E[|m| \; | \; P \text{ succeeds}] +1\\
&\leq (1-p) d \log n + \E[|m|] +1\\
&= (1-p) d \log n + o(d \log n)
\end{align*}
and therefore $P$ cannot exist for $p = 1 - O(1)$.
\end{proof}

\section{Random Graph Reconciliation}
Since reconciliation is harder than graph isomorphism (assuming we determine isomorphic graphs need no changes), we can not expect to develop a protocol that reconciles general graphs in polynomial time.  However, there is a rich line of work showing that for many graphs, and in particular random graphs, graph isomorphism can be solved in polynomial time (with high probability).  This motivates examining one-way graph reconciliation for random graphs; we specifically consider the Erd\H{o}s-R\'enyi ($G(n,p)$) model of random graphs.  Our model is that a base graph $G$ is drawn from the distribution $G(n,p)$, and then Alice and Bob obtain graphs $G_A$ and $G_B$ respectively, where each of $G_A$ and $G_B$ is obtained by making at most $d/2$ edge changes to $G$.  We assume that $p \leq 1/2$, but all our results extend to the case when $p > 1/2$ by taking the graph complement.

Our approach utilizes a methodology used by several graph isomorphism algorithms for random graphs \cite{babai1980random,babai1979canonical,bollobas1998random,czajka2008improved}.  They all find a signature scheme for the vertices where a vertex's signature is invariant under relabeling, and with high probability every vertex in a graph has a unique signature.  When comparing two graphs, if they have different signature sets, we know them to be nonisomorphic.  If their signature sets match, then any labeling must match the vertices with common signatures, which can be used to determine an isomorphism efficiently if one exists.  

To reconcile random graphs, we want a signature scheme with the additional property that the signatures are in some sense robust to a small number of edge changes.  We use this robustness to argue that by reconciling the signatures of $G_A$ and $G_B$, Alice and Bob can agree on a vertex labeling that matches with respect to $G$.  Once we have labeled graphs, complete reconciliation of the edges can be easily done with a standard set reconciliation protocol.  We show that there are existing graph isomorphism schemes for random graphs whose vertex signatures are appropriately robust.

\subsection{Degree Ordering Scheme} \label{sec:gnp_protocol_1}
We start by using the signature scheme of \cite{babai1980random}, which allows a simple protocol;  however, as we later show, we can handle sparser random graphs using a more complex signature scheme and reconciliation protocol.  We first sort the vertices by degree so that $d(v_1) \geq d(v_2) \ldots \geq d(v_n)$.  For the $h$ vertices of largest degree, their signatures are just their degrees.  For the remaining $n-h$ vertices, each vertex's signature, $\sig(v)$, is an $h$-bit string where the $i$th bit of the string, $\sig(v)_i$, denotes whether or not $v$ has an edge to the vertex $v_i$.  The robustness property we utilize is the following.

\begin{definition}
We say a graph is $(h, a, b)$-separated if, after sorting the vertices by degree ($d(v_1) \geq d(v_2) \geq \ldots \geq d(v_n)$), the following properties hold.
\begin{itemize}
\item For all $i \in [h-1]$, $d(v_i)-d(v_{i+1}) \geq a$.
\item For all $i,j \in \{h+1,\ldots,n\}$ and $i \neq j$, the Hamming distance between $\sig(v_i)$ and $\sig(v_j)$, $|\sig(v_i)-\sig(v_j)|$ is at least $b$.
\end{itemize}
\end{definition}

We show below that for an appropriate setting of parameters, $G$ is $(h, d+1, 2d+1)$-separated with high probability.  For now, assume that this is the case.  We claim that if Bob can recover the set of vertex signatures of $G_A$, then the problem reduces to labeled graph reconciliation, which reduces immediately to set reconciliation.  Let $w_1,\ldots,w_n$ be the vertices of $G$.  If a vertex $v_A \in G_A$ and a vertex $v_B \in G_B$ correspond to the same vertex $w_i \in G$, then we say they \emph{conform}.  Thus if we can label $G_A$ and $G_B$ such that every pair of vertices that is labeled the same conform, then we have a labeled graph reconciliation problem with at most $d$ edge differences.  We call this a \emph{conforming labeling}.

Alice labels $G_A$ as follows.  She labels the $h$ highest degree vertices by their degree ordering, and the remaining $n-h$ vertices by the lexicographical order of their $\sig$ strings.  Since $G$ is $(h, d+1, 2d+1)$-separated, and $G_A$ and $G_B$ differ by at most $d$ edges, the $h$ highest degree vertices of $G_A$ will conform to those of $G_B$.  Additionally, for vertices $v_A \in G_A$ and $v_B \in G_B$ which conform and are not in the $h$ highest degrees, $|\sig(v_A)-\sig(v_B)| \leq d$.  If $v_A$ and $v_B$ do not conform, then $|\sig(v_A)-\sig(v_B)| \geq d+1$.  Thus if Bob has all of Alice's signatures, he can change each of his vertices' signatures to the signature of the conforming vertex in $G_A$ and thus match Alice's labeling.

In order for Bob to recover Alice's signatures, we use set of sets reconciliation.  Since the order of the top $h$ degree vertices is the same for Alice and Bob, Alice and Bob need only to reconcile $\sig(v_i)$ for $i > h$.  Each $\sig(v)$ can be interpreted as a subset of $[h]$.  Each edge change only affects the signature of at most one vertex, so the total number of changes across all these sets is at most $h$.  Thus, we can use \autoref{thm:smarter_compacting} to reconcile the signatures and \autoref{thm:setrecon} to reconcile the labeled graphs (in parallel with the signature reconciliation) to obtain the following.

\begin{theorem} \label{thm:deg_order}
If $G=(V,E)$ is $(h, d+1, 2d+1)$-separated and $d$ is known, graph reconciliation can be solved in one round using $O(d(\log d \log h + \log n))$ bits of communication and $O((|E|+d^2)\log d)$ time with probability at least $2/3$.
\end{theorem}

Now we turn to bounding the probability that $G$ is $(h, d+1, 2d+1)$-separated.

\begin{restatable}{theorem}{gnpsep}
\label{thm:gnp_sep}
Given $\delta = \Theta(n^{-\beta})$ with $\beta \in [0,7/68]$, there exist constants $C$ and $N$ such that for $$h = \frac{1}{4}\left(\frac{\delta}{d+1}\right)^{1/3}\left(\frac{p(1-p)n}{\log n}\right)^{1/6},$$ for all $d \geq 2$, $n \geq N$, and $$p \in \left[Cd \log n\left(\frac{d^2}{\delta^2 n}\right)^{1/7},1/2\right],$$ $G(n,p)$ is $(h, d+1, 2d+1)$-separated with probability at least $1 - \delta$.
\end{restatable}

The proof of this theorem appears in \autoref{app:gnp_sep}, and is based on the proof of Lemma 10 in \cite{eppstein2016models}. (We modify and somewhat improve on the parameters for our setting.)  In particular, this theorem means that for constant $d$ our protocol works with constant probability when $p=\Omega(n^{-1/7}\log n)$. 

\subsection{Degree Neighborhood Scheme} \label{sec:gnp_protocol_2}
We can handle sparser graphs using a more complex signature scheme and reconciliation method.  We use the signature scheme of \cite{czajka2008improved}.  This scheme assigns each vertex a signature corresponding to the sorted list of the degrees of its neighbors.  They show this signature allows testing of graph isomorphism to work with high probability for $p \in [\omega(\log^4 n / n \log \log n), 1-\omega(\log^4 n / n \log \log n)]$.  
The robustness guarantee we want for this signature scheme uses the following definition.

\begin{definition}
Let $u$ and $v$ be two different vertices in the graph $G$.  Let $D_u$ be the multiset of the degrees of the vertices in $G$ connected to $u$ whose degrees are at most $m$.  If $|D_u \oplus D_v| \geq k$, we say that $u$ and $v$'s degree neighborhoods are $(m,k)$-\emph{disjoint}.
\end{definition}

\begin{restatable}{theorem}{degdisjoint}
\label{thm:degdisjoint}
For $p \in [\omega(\log^4 n \log \log n / n), 1/2]$ and \\$d = o((pn/\log n)^{3/4}/\log^{1/4}(pn))$, $G(n,p)$'s degree neighborhoods will all be $(pn, 4d+1)$-disjoint with probability at least $1-\exp(-\omega(\log n))$.
\end{restatable}

The proof of this theorem appears in \autoref{app:degdisjoint} and is based on the proof of Theorem 3.8 in \cite{czajka2008improved}, which analyzed the case when $d = 0$.  Assuming $G$'s degree neighborhoods are all $(pn, 4d+1)$-disjoint we can perform graph reconciliation in a manner similar to the protocol used in \autoref{thm:deg_order}.  A vertex $v$'s signature will be $D_v$, the multiset of degrees of vertices connected to $v$ whose degrees are at most $pn$.  Each edge change will change at most $2pn$ vertices' signatures by one or two elements.  Therefore, if $v_A$ and $v_B$ conform then $|D_{v_A} \oplus D_{v_B}| \leq 2d$ and if they do not conform then $|D_{v_A} \oplus D_{v_B}| \geq 2d+1$.  We can then reconcile the graphs by having Bob recover Alice's signatures via set of multisets reconciliation on the signatures.  Bob then matches each of his vertices with a differing signature to the closest signature of Alice's to get a conforming labeling.  They then (in parallel) perform set reconciliation on their labeled vertices.

\begin{theorem}
If $G=(V,E)$'s degree neighborhoods are $(pn,4d+1)$-disjoint and $d$ is known then graph reconciliation can be solved in one round using $O(dpn\log(dpn)\log n)$ bits of communication and 
$O(|E|\log(dpn)+(dpn)^2(d+\log(dpn)))$
time with probability at least $2/3$.
\end{theorem}

\begin{proof}
We use \autoref{thm:smarter_compacting} to reconcile the signatures, using the fact that each of the $d$ edge changes will change at most $O(pn)$ set elements. 
Alice can use Bob's differing signatures to recover a conforming labeling in time $O(d(dpn)^2)$ with probability at least $1 - 1 / \poly(d)$ by computing the size of the difference between signatures using IBLTs.  In parallel, we use \autoref{thm:setrecon} to reconcile the labeled edges.
\end{proof}

This protocol performs significantly worse than \autoref{thm:deg_order} for dense graphs, however it works for much larger ranges of $p$ and $d$.  In particular, it uses roughly $O(pn)$ times as much communication and up to $O(pn)^2$ times as much computation but works for $p$ as small as $O(\log^5 n / n)$ instead of the $p=\Omega(n^{-1/7}\log n)$ required by \autoref{thm:gnp_sep}.

\section{Forest Reconciliation}
Another easy class of graphs for graph isomorphism is that of rooted forests (a collection of rooted trees).  First note that we can easily compute the isomorphism class of a tree as follows.  The label for a vertex is the concatenation of the sorted labels of its children.  The label of the root then indicates the tree's isomorphism class.  This is computable in $O(n)$ time \cite{aho1974design}.  The isomorphism class of a rooted forest is then determined by the sorted list of the trees' labels.  

In \emph{forest reconciliation}, Alice and Bob have rooted forests $G_A$ and $G_B$ such that $G_A$ can be made isomorphic to $G_B$ via at most $d$ directed edge insertions and deletions.  Each of these edge updates must preserve the fact that $G_A$ and $G_B$ are rooted forests.  After a deletion, the child becomes a new root and the child of an inserted edge must have been a root.  Naturally, any deletion will create a new tree and any insertion will connect two trees, making one into a subtree of the other.  
Another way to think of the rooted forest is as a directed forest, with all edges pointing away from the roots.  Any edge deletion will preserve this structure, and edge insertions then create an edge from a vertex to another vertex that was previously a root. 
We wish for Bob to recover a rooted forest isomorphic to $G_A$.

\autoref{thm:com_alg} gave a protocol for general graph reconciliation using $O(d \log n)$ communication whose computation time is dominated by the number of graphs within $d$ operations of $G_A$ times the time required to compute the isomorphism class of a graph.  For general graphs the time to compute the isomorphism class, as far as is currently known, is very expensive. For forest isomorphism there are only $O(n^d)$ possible graphs within $d$ edge changes so the total computation time is $O(n^{d+1})$, which is reasonable for very small $d$.

We develop an improved protocol for the case when none of the trees in $G_A$ and $G_B$ have large depth.  This protocol also uses set of sets reconciliation on vertex signatures.  However, unlike for random graphs, we know of no signature scheme for forests which is robust to even a very small number of edge changes, so our protocol exploits other properties of the signatures to recover $G_A$.
Indeed, this lack of a natural robust signature scheme is what makes this seemingly simple problem appear rather difficult.  

Each vertex's signature is an $\Theta(\log n)$-bit pairwise independent hash of the isomorphism class label of the tree that it roots, with the length of the hash chosen to guarantee that all of the hashes of differing labels will be unique with high probability.  We assume this uniqueness henceforth.  Observe that a forest can be efficiently reconstructed from the multiset of vertex signatures together with the multiset of edge signatures, where an edge signature is simply the ordered pair of appropriate vertex signatures.  For any unique vertex signature we can simply read off the edge signatures to determine its children.  For a vertex signature that occurs $k > 1$ times, the set of edge signatures with it as the parent must be exactly divisible into $k$ identical groups.  We can easily determine these groups, and use that to connect up all of the remaining vertices.

We encode the edge signatures as a multiset of multisets.  Each vertex corresponds to one child multiset, consisting of the vertex's signature (with a special signifier to indicate that it is the parent) together with the signatures of all of its children.  Now observe that a single edge insertion/deletion will only change the signature of at most $\sigma$ different vertices, where $\sigma$ is the maximum depth of a tree in $G_A$ and $G_B$.  Thus, across all of the child multisets at most $O(d \sigma)$ changes occur, so we can efficiently reconcile them using the protocols from \autoref{sec:sos}.

\begin{theorem}
If $d$ is known, forest reconciliation can be solved in one round using $$O(d\sigma\log(d\sigma) \log n)$$ bits of communication and $$O((n+(d\sigma)^2)\log(d\sigma))$$ time with probability at least $2/3$.
\end{theorem}

\begin{proof}
Each party computes the signature of each of their vertices as follows.  Every vertex's signature is an $O(\log n)$-bit pairwise independent hash of the sorted list of its children's signatures, except leaves' signatures which are the hash of $0$.  All of the vertices' signatures can be computed in $O(n)$ time if we sort signatures using the radix sort, since the signatures are $O(1)$ words each.  We encode the edge signatures as multisets of vertex signatures, as previously described, and reconcile them using \autoref{thm:smarter_compacting}, using the fact that at most $O(d \sigma)$ changes occur across all of the multisets.

Once Bob has Alice's vertex and edge signatures he can recover $G_A$ in $O(n)$ time.  He first creates a vertex for each vertex signature and then radix sorts the list of edge signatures in lexicographical order (with the parent signature first).  For each edge, if the parent signature appears only once in the list of vertex signatures he adds an edge from the lone vertex with that signature to any root vertex with the child's signature (if there are multiple candidate children it doesn't matter which he picks since the subtrees are isomorphic).  If there are $k$ vertices that match the parent's signature, then there must be $ck$ copies of that edge signature for some integer $c$.  Bob then simply adds $c$ of these edges to each of the $k$ parents.
\end{proof}

\section{Conclusion}
We have introduced the problem of set of sets reconciliation, as well as several solutions for it.  We demonstrated that it provides a useful primitive for graph reconciliation, and described its use for other problems.  Several natural questions remain, including finding efficient graph reconciliation algorithms for additional classes of graphs, and determining tight upper and lower bounds for set of sets reconciliation problems.  A further direction would be to determine the effectiveness of various set of sets reconciliation algorithms in a practical application, with an eye towards understanding what bottlenecks exist in practice and how that might affect the design of such algorithms.  

\section*{Acknowledgements}
We would like to thank Justin Thaler for his valuable discussions about graph reconciliation.

Michael Mitzenmacher was supported in part by NSF grants CNS-1228598, CCF-1320231, CCF-1563710 and CCF-1535795.  Tom Morgan was supported in part by NSF grants CNS-1228598 and CCF-1320231.

\bibliographystyle{plain}
\bibliography{approx_set_recon}

\begin{thebibliography}{10}

\bibitem{rsync}
rsync.
\newblock {\em \url{https://rsync.samba.org}}.

\bibitem{aho1974design}
Alfred~V. Aho, John~E. Hopcroft, and Jeffrey~D. Ullman.
\newblock {\em The Design and Analysis of Computer Algorithms}.
\newblock Addison-Wesley Pub. Co., 1974.

\bibitem{babai2016graph}
L{\'a}szl{\'o} Babai.
\newblock Graph isomorphism in quasipolynomial time.
\newblock In {\em Proceedings of the 48th Annual ACM SIGACT Symposium on Theory
  of Computing}, pages 684--697. ACM, 2016.

\bibitem{babai1980random}
L{\'a}szl{\'o} Babai, Paul Erd{\"o}s, and Stanley Selkow.
\newblock Random graph isomorphism.
\newblock {\em SIAM Journal on Computing}, 9(3):628--635, 1980.

\bibitem{babai1979canonical}
L{\'a}szl{\'o} Babai and Ludik Kucera.
\newblock Canonical labelling of graphs in linear average time.
\newblock In {\em Proc. of the 20th Annual Symposium on Foundations of Computer
  Science}, pages 39--46, 1979.

\bibitem{bille2005survey}
Philip Bille.
\newblock A survey on tree edit distance and related problems.
\newblock {\em Theoretical Computer Science}, 337(1):217--239, 2005.

\bibitem{bollobas1998random}
B{\'e}la Bollob{\'a}s.
\newblock {\em Random graphs}.
\newblock Springer, 1998.

\bibitem{boral2014multi}
Anudhyan Boral and Michael Mitzenmacher.
\newblock Multi-party set reconciliation using characteristic polynomials.
\newblock In {\em Proceedings of the 52nd Annual Allerton Conference on
  Communication, Control, and Computing}, pages 1182--1187, 2014.

\bibitem{broder1997resemblance}
Andrei Broder.
\newblock On the resemblance and containment of documents.
\newblock In {\em Proceedings of Compression and Complexity of Sequences},
  pages 21--29, 1997.

\bibitem{brodnik1993computation}
Andrej Brodnik.
\newblock Computation of the least significant set bit.
\newblock In {\em Proceedings of the 2nd Electrotechnical and Computer Science
  Conference}, 1993.

\bibitem{czajka2008improved}
Tomek Czajka and Gopal Pandurangan.
\newblock Improved random graph isomorphism.
\newblock {\em Journal of Discrete Algorithms}, 6(1):85--92, 2008.

\bibitem{eppstein2011straggler}
David Eppstein and Michael Goodrich.
\newblock Straggler identification in round-trip data streams via {N}ewton's
  identities and invertible {B}loom filters.
\newblock {\em IEEE Transactions on Knowledge and Data Engineering},
  23(2):297--306, 2011.

\bibitem{eppstein2016models}
David Eppstein, Michael Goodrich, Jenny Lam, Nil Mamano, Michael Mitzenmacher,
  and Manuel Torres.
\newblock Models and algorithms for graph watermarking.
\newblock In {\em Proc. of the International Conference on Information
  Security}, pages 283--301, 2016.

\bibitem{eppstein2011s}
David Eppstein, Michael Goodrich, Frank Uyeda, and George Varghese.
\newblock What's the difference?: efficient set reconciliation without prior
  context.
\newblock {\em ACM SIGCOMM Computer Communication Review}, 41(4):218--229,
  2011.

\bibitem{fredman1993surpassing}
Michael Fredman and Dan Willard.
\newblock Surpassing the information theoretic bound with fusion trees.
\newblock {\em Journal of Computer and System Sciences}, 47(3):424--436, 1993.

\bibitem{gm11}
Michael Goodrich and Michael Mitzenmacher.
\newblock Invertible {B}loom lookup tables.
\newblock In {\em Proceedings of the 49th Annual Allerton Conference on
  Communication, Control, and Computing}, pages 792--799, 2011.

\bibitem{greenberg2014tight}
Spencer Greenberg and Mehryar Mohri.
\newblock Tight lower bound on the probability of a binomial exceeding its
  expectation.
\newblock {\em Statistics \& Probability Letters}, 86:91--98, 2014.

\bibitem{irmak2005improved}
Utku Irmak, Svilen Mihaylov, and Torsten Suel.
\newblock Improved single-round protocols for remote file synchronization.
\newblock In {\em IEEE INFOCOM 2005}, 2005.

\bibitem{kane2010optimal}
Daniel Kane, Jelani Nelson, and David Woodruff.
\newblock An optimal algorithm for the distinct elements problem.
\newblock In {\em Proc. of the Twenty-Ninth ACM SIGMOD-SIGACT-SIGART Symposium
  on Principles of Database Systems}, pages 41--52, 2010.

\bibitem{kazemi2015growing}
Ehsan Kazemi, S.~Hamed Hassani, and Matthias Grossglauser.
\newblock Growing a graph matching from a handful of seeds.
\newblock {\em Proceedings of the VLDB Endowment}, 8(10):1010--1021, 2015.

\bibitem{kazemi2015can}
Ehsan Kazemi, Lyudmila Yartseva, and Matthias Grossglauser.
\newblock When can two unlabeled networks be aligned under partial overlap?
\newblock In {\em Proceedings of the 53rd Annual Allerton Conference on
  Communication, Control, and Computing}, pages 33--42. IEEE, 2015.

\bibitem{korula2014efficient}
Nitish Korula and Silvio Lattanzi.
\newblock An efficient reconciliation algorithm for social networks.
\newblock {\em Proceedings of the VLDB Endowment}, 7(5):377--388, 2014.

\bibitem{minsky2003set}
Yaron Minsky, Ari Trachtenberg, and Richard Zippel.
\newblock Set reconciliation with nearly optimal communication complexity.
\newblock {\em IEEE Transactions on Information Theory}, 49(9):2213--2218,
  2003.

\bibitem{mitzenmacher2013simple}
Michael Mitzenmacher and Rasmus Pagh.
\newblock Simple multi-party set reconciliation.
\newblock {\em arXiv preprint arXiv:1311.2037}, 2013.

\bibitem{newman1991private}
Ilan Newman.
\newblock Private vs. common random bits in communication complexity.
\newblock {\em Information Processing Letters}, 39(2):67--71, 1991.

\bibitem{schwartz1980fast}
Jacob Schwartz.
\newblock Fast probabilistic algorithms for verification of polynomial
  identities.
\newblock {\em Journal of the ACM (JACM)}, 27(4):701--717, 1980.

\bibitem{starobinski2003efficient}
David Starobinski, Ari Trachtenberg, and Sachin Agarwal.
\newblock Efficient pda synchronization.
\newblock {\em IEEE Transactions on Mobile Computing}, 2(1):40--51, 2003.

\bibitem{rsyncalg}
Andre Trigdell and Paul Mackerras.
\newblock The rsync algorithm.
\newblock {\em \url{https://rsync.samba.org/tech_report/}}.

\bibitem{yan2008algorithms}
Hao Yan, Utku Irmak, and Torsten Suel.
\newblock Algorithms for low-latency remote file synchronization.
\newblock In {\em IEEE INFOCOM 2008}, pages 156--160. IEEE, 2008.

\bibitem{yartseva2013performance}
Lyudmila Yartseva and Matthias Grossglauser.
\newblock On the performance of percolation graph matching.
\newblock In {\em Proceedings of the 1st Conference on Online Social Networks
  (COSN)}, number EPFL-CONF-189760. ACM, 2013.

\bibitem{zippel1979probabilistic}
Richard Zippel.
\newblock Probabilistic algorithms for sparse polynomials.
\newblock {\em Symbolic and Algebraic Computation}, pages 216--226, 1979.

\end{thebibliography}

\appendix

\section{Set Difference Estimators}
We construct a set difference estimator making use of streaming approximations for the $\ell_0$-norm.  We consider a vector where the number of dimensions corresponds to the number of possible set elements.  We use the work of \cite{kane2010optimal} on the RoughL0Estimator to prove the following.

\label{app:strata}
\begin{theorem} \label{thm:stream}
There is a streaming algorithm for the turnstile model using $O(\log n)$ space for streams of length $n$ such that when every dimension is in $\{-1,0,1\}$, with probability at least 9/16 the algorithm outputs an 110-approximation to the $\ell_0$-norm.  The time to update, query, or merge two sketches produced by the algorithm is $O(1)$.
\end{theorem}

\autoref{thm:strata} follows from \autoref{thm:stream} for $\delta = O(1)$ as follows.  The set difference estimator is the sketch.  To add an element to $S_1$ corresponds to a $+1$ update to the corresponding dimension in the stream.  Adding an element to $S_2$ corresponds to a $-1$ update.  The merge operation is performed by simply merging the sketches, and querying for the set $\ell_0$ estimate gives us our set difference estimate since the $\ell_0$-norm of the difference between the sets' indicator vectors is exactly the size of the set difference.
The standard approach of taking the median of $O(\log (1 / \delta))$ parallel runs of the algorithm lowers the failure probability to $\delta$.

Appendix A and Theorem 11 of \cite{kane2010optimal} shows that with an appropriate subroutine, the space and probability bounds of \autoref{thm:stream} are achieved.  The subroutine required must have the property that when given the promise that the $\ell_0$-norm is at most $c$, it outputs the $\ell_0$-norm exactly with probability $1-\eta$ and takes $O(c^2)$ space.  The original subroutine provided in Lemma 8 of \cite{kane2010optimal} took more space than this, but was more general in that it applied to cases where the dimensions were not limited to being in $\{-1,0,1\}$.

Given this subroutine, their algorithm is as follows.  The $u$ dimensions are partitioned into $\log n$ groups, where the probability of a dimension falling into the $i$th group is $1/2^i$, using a pairwise independent hash function from the dimension to $[n]$ (uniformly), where the group is the least significant bit of a dimension's hash.  For each $i \leq \log n$ they consider the substream $S_i$ consisting of updates to dimensions in group $i$.  On each substream they run an instantiation of the subroutine with $c = 141$ and $\eta = 1/16$.  The final estimate of $\ell_0$ is then $2^{i^*}$, where $i^*$ is the largest value of $i$ whose subroutine reports that the $\ell_0$-norm is greater than 8.

Now, our subroutine simply hashes the universe into $\Theta(c^2)$ buckets,
and stores a 2 bit counter for each bucket, maintaining the sum of all of the elements that hash to that bucket modulo 4.  With constant probability, none of the at most $c$ non-zero dimensions hash to the same bucket, and as each dimension is in $\{-1,0,1\}$ the number of non-zero buckets will equal the $\ell_0$-norm.  To amplify this to probability $1-\eta$ we replicate the subroutine $O(\log(1/\eta))$ times and take the maximum result.  For $c = O(1)$ and $\eta = O(1)$ the subroutine requires $O(1)$ space per group, so the total space needed is $O(\log n)$ 
as desired.

We improve on the query/merge time of Theorem 11 of \cite{kane2010optimal} for our special case of all dimensions in $\{-1,0,1\}$.  Since our subroutine takes a constant number of bits, the whole data structure fits in a constant number of words and we can exploit the power of word RAM to operate on them efficiently.  Our sketch is simply $\log n = O(w)$ levels of $O(c^2 \log (1 / \eta)) = O(1)$ two bit counters per level.  To merge two sketches we simply need to merge all of these buckets.  If they were not all mod 4, we could just add together the sketches as $O(1)$ $w$-bit integers and be done.  Instead, we will store each bucket as three bits instead of two using one padding bit between buckets that will always be zero.  Now to add two buckets together mod 4 we just add the buckets then zero the padding bit.  Thus to merge the sketches we just add them together as $O(1)$ $w$-bit integers and apply a mask to each word to zero all the padding bits.  Since the whole sketch is $O(1)$ words long, this takes $O(1)$ time.

To query in $O(1)$ time, we are going to leverage the fact that we can compute the least significant bit in a word in $O(1)$ time \cite{brodnik1993computation, fredman1993surpassing}.  We are going to use $O(1)$ operations to compute in parallel for all $\log n$ levels a single bit indicating whether or not the subroutine for that level reports an estimate greater than 8.  We then take the least significant bit (or most significant bit depending on the orientation) to find the largest level with a 1, and that gives us our report.

\section{Multi-Round Protocol}
\label{app:multi}

Here we prove the correctness of the the multi-round protocol described in \autoref{sec:multi}.

\multi*

\begin{proof}
Excluding some replication for probability amplification, the protocol follows.  We argue its correctness afterwards.
\begin{enumerate}
\item Alice computes an $O(\log s)$-bit pairwise independent hash of her child sets and inserts all of the hashes into an $O(\q)$-cell IBLT $T_A$ which she transmits to Bob.
\item Bob computes an $O(\log s)$-bit pairwise independent hash of his child sets and inserts all of the hashes into an $O(\q)$-cell IBLT $T_B$.  Bob decodes $(T_A,T_B)$, and determines which of his child sets differ from Alice.  For each of his differing child sets, he constructs a set difference estimator and puts all of these estimators into a list $L_B$.  He transmits $T_B$ and $L_B$ to Alice.
\item Alice decodes $(T_A,T_B)$, and constructs $L_A$, a list of set difference estimators for each of her differing child sets.  For each set difference estimator $L_{A,i} \in L_A$ and $L_{B,j} \in L_B$, she merges them and estimates the difference.  Let $b_i$ be the index $j$ of the $L_{B,j}$ with which $L_{A,i}$ reported the smallest difference, and let $d_i$ be that reported difference.  For each $i$ with $d_i \geq \sqrt{d}$, Alice transmits $b_i$ along with $T_i$, an $O(d_i)$-cell IBLT of the child set corresponding to $L_{A,i}$.  For each $i$ such that $d_i < \sqrt{d}$, Alice transmits $b_i$ together with a list $P_i$ of $O(d_i)$ evaluations of the characteristic polynomial of the child set corresponding to $L_{A,i}$.
\item For each of the received $(b_i, T_i)$ pairs, Bob recovers Alice's set by deleting the elements of his child set corresponding to $b_i$ from $T_i$, decoding it and then applying the extracted differences to the child set.  For each of the $(b_i, P_i)$ pairs, Bob recovers Alice's set as in \autoref{thm:setreconpoly}, by computing the evaluations of the characteristic polynomial of the child set corresponding to $b_i$, finding the roots of the rational expression, and applying those as differences to his child set.  Bob then recovers Alice's total set of sets by removing all child sets corresponding to $L_B$ from his set and adding in Alice's child sets that he has recovered. 
\end{enumerate}

This protocol succeeds so long as none of the hashes collide, $T_A$ and $T_B$ together decode, none of the queried pairs of set difference estimators fail, and all of the $T_i$s decode.  None of the hashes collide with probability at least $1 - 1 / \poly(s)$.  $(T_A,T_B)$ decodes with probability at least $1 - \poly(\q)$.  By replicating step 1 (and the corresponding part of step 2) $O(\lceil\log_{\q}(1/\delta)\rceil)$ times, we reduce the probability that any of the hashes collide or $T_A$ and $T_B$ fails to decode to $\delta/2$.

There are $O(\q^2)$ pairs of set difference estimators, so we choose the failure probability of each one in \autoref{thm:strata} to be $\delta / \poly(\q)$ so that they all succeed with probability at least $1 - \delta / 4$.  There are $O(\sqrt{d})$ $T_i$s, and each one decodes with probability at least $1 - 1/\poly(\sqrt{d})$, so by choosing the constants in the IBLT appropriately, and then replicating each one $O(\lceil\log_d(1/\delta)\rceil)$ times, they decode with probability at least $1 - \delta/4$.  Putting it all together, the protocol succeeds with probability at least $1 - \delta$.

Constructing and decoding $T_A$ and $T_B$ takes, over $O(\lceil\log_{\q}(1/\delta)\rceil)$ replications, $O(\lceil\log_{\q}(1/\delta)\rceil n)$ time and $O(\lceil\log_{\q}(1/\delta)\rceil\q\log s)$ bits of communication.  By \autoref{thm:strata}, constructing $L_A$ and $L_B$ takes $O(\log (\q / \delta) n)$ time and transmitting $L_B$ $O(\log (\q / \delta)\q \log h)$ bits of communication.  Finding the $b_i$s and $d_i$s consists of $O(\q^2)$ set difference merges and queries, which by \autoref{thm:strata} take a total of $O(\log (\q / \delta)\q^2)$ time.  Sending the $b_is$ takes $O(\q\log \q)$ bits of communication.  Computing and decoding the $T_i$s takes $O(\lceil\log_d(1/\delta)\rceil n)$ time and transmitting them takes $O(\lceil\log_d(1/\delta)\rceil d \log u)$ bits of communication.  By \autoref{thm:setreconpoly} constructing and decoding the $P_i$s takes time $O(d^2+\min(d h, n \sqrt{d}, n \log^2 h))$.  Transmitting the $P_i$s takes $O(d \log u)$ bits of communication.  These terms all add up to our desired communication cost and computation time.
\end{proof}

\multiu*

\begin{proof}
We proceed as in \autoref{thm:strata_based}, except that before step 1, Bob sends Alice a set difference estimator for the set of his child hashes, which Alice uses to estimate the number of cells need in $T_A$.  By \autoref{thm:strata}, this set difference estimator takes $O(\log(1/\delta)n)$ time and $O(\log(1/\delta) \log s)$ bits of communication.
\end{proof}

\section{Degree Ordering Separation Bound}
\label{app:gnp_sep}

Here we will prove the following theorem, which we used in \autoref{sec:gnp_protocol_1} for reconciling random graphs.  Recall that a graph is $(h, a, b)$-separated if the $h$ highest degree vertices' degrees are all separated by at least $a$, and the remaining vertices' connectivity bit vectors with the top $h$ vertices are all separated (in Hamming distance) by at least $b$.

\gnpsep*

To prove this, we need the following lemma, which upper bounds the probability that the $m$ highest degree vertices are not well separated.  Here $d_i$ is the degree of the vertex with the $i$th highest degree in the graph.

\begin{lemma}[\cite{eppstein2016models} Lemma 9]  
\label{lem:watermarking_sep}

Let $\alpha(n) \rightarrow 0$, $m = o(p(1-p)n/\log n)^{1/4}$, and $m \rightarrow \infty$.  Then $G(n,p)$ has for all $i < m$,
$$d_i - d_{i+1} \geq \frac{\alpha(n)}{m^2}\sqrt{\frac{p(1-p)n}{\log n}}$$
with probability at least
$$1 - m \alpha(n) - \frac{1}{m \log^2(n/m)}.$$
\end{lemma}

With this lemma in hand, we can prove our theorem.  We use the lemma directly to prove the graph is $(h,d+1,0)$-separated, and then apply a Chernoff bound to argue the graph is $(h,0,2d+1)$ separated.

\begin{proof}[Proof of \autoref{thm:gnp_sep}]  Let $$m = h = \frac{1}{4}\left(\frac{\delta}{d+1}\right)^{1/3}\left(\frac{p(1-p)n}{\log n}\right)^{1/6}$$ and $$\alpha(n) = m^2(d+1)\sqrt{\frac{\log n}{p(1-p) n}}.$$
We now argue that we meet all the conditions required to use \autoref{lem:watermarking_sep}.  First, $m = o(p(1-p)n/\log n)^{1/4}$ because $\frac{\delta}{d} < 1$.

Since $p = \Omega(d^{9/7}n^{(2\beta-1)/7}\log n)$, we have that $d = O(n^{(1-2\beta)/9}p^{7/9})$.  Therefore,
\begin{align*}
m &= \Omega\left(\left(\frac{n^{-\beta}}{n^{(1-2\beta)/9}p^{7/9}}\right)^{1/3} \left(\frac{pn}{\log n}\right)^{1/6}\right) \\
&= \Omega\left(\frac{n^{\frac{7}{54}(1-2\beta)}}{p^{5/54}\log^{1/6} n}\right) \\
&= \omega(1),
\end{align*}
because $\beta < 1/2$, and thus $m \rightarrow \infty$.

Also,
\begin{align*}
\alpha(n) &= \Theta\left(m^2 d \sqrt{\frac{\log n}{pn}}\right) \\
&= \Theta\left(\left(\frac{\delta}{d}\right)^{2/3}\left(\frac{pn}{\log n}\right)^{1/3} d \sqrt{\frac{\log n}{pn}}\right) \\
&= \Theta\left((\delta^2d)^{1/3}\left(\frac{\log n}{pn}\right)^{1/6} \right) \\
&= O\left(n^{(-2\beta+(1-2\beta)/9)/3}p^{5/54}\left(\frac{\log n}{n}\right)^{1/6}\right) \\
&= O\left(n^{-\frac{1}{54}(7+40\beta)}p^{5/54}\log^{1/6} n\right) \\
&= o(1),
\end{align*}
and thus $\alpha(n) \rightarrow 0$.

Finally,
\begin{align*}
1& - m \alpha(n) - \frac{1}{m \log^2(n/m)} \\
&= 1 - \delta/4 - O\left(\frac{1}{m\log^2\left(n\left(\frac{d}{\delta}\right)^{1/3}\left(\frac{\log n}{pn}\right)^{1/6}\right)}\right) \\
&= 1 - \delta/4 - O\left(\frac{n^{-\beta}}{\log^{11/6}\left(n\right)}\right) \\
&= 1 - \delta/4 - o(\delta).
\end{align*}
The second to last equality used the fact that $\beta \leq 7/68$.  We now have that by \autoref{lem:watermarking_sep}, for sufficiently large $n$, $d_i - d_{i+1} \geq d+1$ for all $i < h$ with probability at least $1 - \delta/2$.

For any vertex $v \in G$, $\sig(v)$ is an $h$-bit string where the $j$th bit of the string, $\sig(v)_j$, denotes whether or not $v$ has an edge to $v_j$, the vertex with the $j$th highest degree in $G$.  We wish to bound $X_{u,v}=|\sig(u)-\sig(v)|$ (the Hamming distance) for all pairs of vertices $u \neq v \in G$.  Unfortunately, the values $|\sig(u)_j-\sig(v)_j|$ are not quite distributed as i.i.d. Bernoulli random variables (which would be easy to analyze) because we are conditioning on $v_j$ having the $j$th highest degree.

To remedy this, we introduce $G'$, which is $G$ without the edges adjacent to $u$ and $v$.  Let $\sig'(v)$ be an $h'=h-2$-bit string where $\sig'(v)_j$ denotes whether or not $v$ has an edge in $G$ to $v'_j$, the vertex with the $j$th highest degree in $G'$.  Let $X'_{u,v}=|\sig'(u)-\sig'(v)|$.  Since the edges adjacent to $u$ and $v$ are not used in ordering the degrees for determining $v'_j$, $X'_{u,v}$ is distributed as the sum of $h'$ i.i.d. Bernoulli random variables, each of which is 1 with probability $2p(1-p)$.  It suffices to bound $X'_{u,v}$ because, conditioned on the event that $d_i - d_{i+1} \geq 3$ for all $i < h$, if $X'_{u,v} > 2d$ then $X_{u,v} > 2d$.  This is because the top $h'$ highest degree vertices of $G'$ must still be in the top $h$ highest degree vertices of $G$.  

We now have
$$\E[X'_{u,v}] = 2h'p(1-p) = (2(1-p)) h' p \geq h' p$$
as $p \leq 1/2$.  
For sufficiently large $n$,
\begin{align*}
h'p &= \left(\frac{1}{4}\left(\frac{\delta}{d+1}\right)^{1/3}\left(\frac{p(1-p)n}{\log n}\right)^{1/6}-2\right) p \\
&\geq \frac{1}{8}\left(\frac{\delta}{d}\right)^{1/3}\left(\frac{n}{\log n}\right)^{1/6} p^{7/6} \\
&\geq \frac{1}{8}\left(\frac{\delta}{d}\right)^{1/3}\left(\frac{n}{\log n}\right)^{1/6} \left(Cd \log n\left(\frac{d^2}{\delta^2 n}\right)^{1/7}\right)^{7/6} \\
&= \frac{(Cd)^{7/6}}{8} \log n,
\end{align*}
and thus for sufficiently large $n$, $h' p \geq 4d$.  By a Chernoff bound,
\begin{align*}
\Pr[X'_{u,v} \leq 2d] &\leq \exp\left(\frac{-(h' p-2d)^2}{2h'p}\right) \\
&\leq \exp\left(\frac{-(\frac{1}{2}h' p)^2}{2h'p}\right) = \exp\left(\frac{-h' p}{8}\right) \\
&\leq \exp\left(\frac{-(Cd)^{7/6} \log n}{64}\right) \leq n^{-C^{7/6}/32},
\end{align*}
which is at most $\delta / (2n^2) = \Theta(n^{-\beta-2})$ for sufficiently large $C$ and $n$.  By union bounding over all $O(n^2)$ vertices $u$ and $v$, we have the theorem.
\end{proof}

\section{Degree Neighborhood Disjointness Bound}
\label{app:degdisjoint}

In this section we prove a lower bound on the probability that $G(n,p)$'s degree neighborhoods are $(pn, 4d+1)$-disjoint.  Recall that two vertices $u$ and $v$ have degree neighborhoods that are $(m,k)$ disjoint if the multiset of degrees less than $m+1$ connected to $u$ has at least $k$ elements different from the multiset of degrees less than $m+1$ connected to $v$.

To prove our bound, we will need the following lemma which allows us to remove $u$ and $v$'s contribution to their neighbors' degrees from consideration when analyzing their degree neighborhood disjointness.  This is important as it will allow us to treat the edges out of $u$ and $v$ as independent of the degrees of the vertices they connect to.

\begin{lemma}[Extension of \cite{czajka2008improved} Theorem 3.5]
\label{lem:indep_disjoint}
Let $u$ and $v$ be two different vertices in $G$.  Let $G'$ be the subgraph of $G$ formed by removing $u$ and $v$.  Suppose the multiset of the $G'$-degrees of the vertices in $G'$ connected to $u$ whose $G'$-degrees are at most $m-1$ has $k+2$ elements different from the multiset of the $G'$-degrees of the vertices in $G'$ connected to $v$ whose $G'$-degrees are at most $m-1$.  Then $u$ and $v$'s degree neighborhoods in $G$ are $(m,k)$-disjoint.
\end{lemma}
\begin{proof}
Let $D_x$ be the multiset of the $G$-degrees of the vertices in $G$ connected to $x$ whose $G$-degrees are at most $m$.  Let $D_x'$ be the multiset of the $G'$-degrees of the vertices in $G'$ connected to $x$ whose $G'$-degrees are at most $m-1$.  Let $D_x''$ be the multiset of the $G$-degrees of the vertices in $G'$ connected to $x$ whose $G$-degrees are at most $m$.

First we will argue that $|D''_u \oplus D''_v| = |D'_u \oplus D'_v|$.  Let $M$ be a maximum matching between the set of vertices of $G'$ connected to $u$ whose $G'$-degrees are at most $m-1$ and the set of vertices of $G'$ connected to $v$ whose $G'$-degrees are at most $m-1$ with the following properties.  In order to be matched, two vertices must have the same degree.  If $u$ and $v$ are connected to the same vertex, that vertex must be matched to itself.  Let $U_u$ be $u$'s unmatched vertices and $U_v$ be $v$'s unmatched vertices.  Observe that $|D'_u \oplus D'_v| = |U_u| + |U_v|$, $|D'_u| = |M| + |U_u|$ and $|D'_v| = |M| + |U_v|$.

Consider how $M$ changes when we instead look at the $G$-degrees of the vertices and allow degrees up to $m$ instead of $m-1$.  Every one of the relevant vertices' degrees increases by either 1 or 2.  If the vertex is adjacent to both $u$ and $v$, its degree increases by 2 and it will either stay matched, or exceed degree $m$.  Every other vertex's degree increases by 1.  The set of these vertices that had degree at most $m-1$ before is exactly the same as the set that now has degree at most $m$.  All of the degrees increase by exactly one, so the matches are still valid matches and the unmatched can still not be matched.  Thus, $|D''_u \oplus D''_v| = |U_u| + |U_v| = |D'_u \oplus D'_v|$.

Now note that $|D_u \oplus D_v| \geq |D''_u \oplus D''_v| - 2$.  This is because if $u$ and $v$ are not connected, then $D''_u = D_u$ and $D''_v = D_v$.  If they are connected, then each multiset increases by at most 1 and thus the difference decreases by at most 2.  Therefore, we have 
$$|D_u \oplus D_v| \geq |D''_u \oplus D''_v| - 2 = |D'_u \oplus D'_v| - 2 = k.$$
\end{proof}

The following lemma says that $G(n,p)$ has many disjoint ranges of degrees (\emph{degree ranges}) in which many vertices fall.  We use this to relate the degree neighbor disjointness of two vertices $u$ and $v$ to the difference in number of connections $u$ and $v$ have to the sets of vertices in each of these degree ranges.   Each of these differences are distributed as the difference of independent binomial random variables.

\begin{lemma}[\cite{czajka2008improved} Corollary 3.7] 
\label{lem:heavy_buckets} 
If $p \in [\omega(\log^3 n / n), 1/2]$, for any $x = o\left(\sqrt{pn/(\log^2(n)\log(pn))}\right)$ there exists $R = \omega\left(\sqrt{\log n \log(pn)}\right)$ such that in $G(n,p)$ there will be at least $K = \left\lfloor\sqrt{n \log n \log(pn)/p} \right\rfloor$ vertices in each of the degree ranges $(pn-xR-1,pn-(x+1)R-1],(pn-(x+1)R-1,pn-(x+2)R-1],\ldots,(pn-R-1,pn-1]$ with probability at least $1 - \exp(-\omega(\log n))$.
\end{lemma}

Finally we use the following technical lemma regarding the difference of binomial random variables.  This may be folklore, but we provide a proof for completeness.

\begin{lemma}
\label{lem:binomial_deviation}
Let $X$ and $Y$ be i.i.d. random variables distributed according to $\mathrm{Bin}(n,p)$ for $p \in [\omega(1/n), 1/2]$. Then
$$\E[\min\{|X-Y|, \sqrt{np}\}] = \Omega(\sqrt{np}).$$
\end{lemma}
\begin{proof}
Let $b(k; n, p) = \Pr[\mathrm{Bin}(n,p) = k]$.  By Lemma 3.1 of \cite{czajka2008improved} (or Stirling's approximation), $b(\lfloor np \rfloor; n, p) = \Theta(1/\sqrt{np})$.  For all $k \in [np-2\sqrt{np},np]$,
\begin{align*}
\frac{b(k; n, p)}{b(k-1; n, p)} &= \frac{p}{1-p} \cdot \frac{n-k+1}{k-1} \\
&= 1 + \frac{np+p-k}{k(1-p)} \\
&\leq 1 + \frac{np+p-(np-2\sqrt{np})}{(np-2\sqrt{np})(1-p)} \\
&\leq 1 + \frac{2(p + 2\sqrt{np})}{np - 2\sqrt{np}} \\
&= 1 + O\left(\frac{1}{\sqrt{np}}\right).
\end{align*}
This implies that
$$\frac{b(\lfloor np \rfloor; n, p)}{b(\lfloor np -2\sqrt{np}\rfloor; n, p)} \leq \left(1 + O\left(\frac{1}{\sqrt{np}}\right) \right)^{\lceil 2\sqrt{np}\rceil} = O(1),$$
and since $b(k-1;n,p) < b(k;n,p)$ for $k \leq \lfloor np \rfloor$,
\begin{align*}
\Pr[X \leq np - \sqrt{np}] &\geq \Pr[X \in [np - 2\sqrt{np}, np-\sqrt{np}]] \\
&\geq \sqrt{np} \cdot \Pr[X = \lfloor np - 2\sqrt{np} \rfloor] \\
&\geq \sqrt{np} \cdot \Omega(1) \cdot \Pr[X = \lfloor np \rfloor] \\
&= \sqrt{np} \cdot \Omega(1) \cdot \Theta(1 / \sqrt{np}) \\
&= \Omega(1).
\end{align*}

By Theorem 1 of \cite{greenberg2014tight} $\Pr[Y \geq np] > 1/4$, and thus
\begin{align*}
\E&[\min\{|X-Y|, \sqrt{np}\}] \\
&\geq \sqrt{np} \cdot \Pr[X \leq np - \sqrt{np}] \cdot \Pr[Y \geq np] \\
&\geq \sqrt{np} \cdot \Omega(1) \cdot (1/4) \\
&= \Omega(\sqrt{np}).
\end{align*}
\end{proof}

Putting all these lemmas together, we now prove our theorem.

\degdisjoint*

\begin{proof}
Let $n'=n-2$, $K = \left\lfloor\sqrt{n' \log n' \log(pn')/p} \right\rfloor$, and $x$ be a value specified later that is $o\left(\sqrt{pn/(\log^2(n)\log(pn))}\right)$ .  

Consider two different vertices $u$ and $v$ in the graph $G$.  Let $G'$ be the subgraph of $G$ formed by removing $u$ and $v$.  $G'$ is a random $G(n',p)$ graph, so by \autoref{lem:heavy_buckets}, with high probability there are $x$ disjoint degree ranges with at least $K$ vertices in each range as given by the theorem.

Let $U_i$ be a random variable denoting the number of vertices in the $i$th of these degree ranges that $u$ is adjacent to.  Similarly, let $V_i$ be a random variable denoting the number of vertices in the $i$th of these degree ranges that $v$ is adjacent to.   Observe that the $U_i$s and $V_i$s are all independent binomial random variables with at least $K$ trials, and probability of success $p$.  This implies that if $X_1,\ldots,X_x,Y_1,\ldots,Y_x$ are i.i.d. random variables distributed according to $\mathrm{Bin}(K,p)$, then
\begin{align*}
\Pr&\left[\sum_{i=1}^x|U_i-V_i|\leq 4d+2\right] \\
&\leq \Pr\left[\sum_{i=1}^x|X_i-Y_i|\leq 4d+2\right] \\
&\leq \Pr\left[\sum_{i=1}^x\min\{|X_i-Y_i|,\sqrt{Kp}\}\leq 4d+2\right].
\end{align*}

Let $\alpha = \E[\min\{|X_i-Y_i|,\sqrt{Kp}\}] / \sqrt{Kp}.$  Assuming $4d+2 \leq \alpha x \sqrt{Kp}$, by a Chernoff bound,
\begin{align*}
\Pr&\left[\sum_{i=1}^x\min\{|X_i-Y_i|,\sqrt{Kp}\}\leq 4d+2\right] \\
&\leq \exp\left(\frac{-(\alpha x - (4d+2)/\sqrt{Kp})^2}{2\alpha x}\right).
\end{align*}

Now we argue that we can choose $x$ such that $x = \omega(\log n)$ and $\alpha x / 2 \geq (4d+2)/\sqrt{Kp}$ while still maintaining \\$x = o\left(\sqrt{pn/(\log^2(n)\log(pn))}\right)$.  For the first part, observe that
$$\frac{pn}{\log(pn)} = \omega\left(\log^4 n\right)$$
which implies that
$$\log n = o\left(\sqrt{\frac{pn}{\log^2n \log(pn)}}\right).$$

For the second part, note that
\begin{align*}
(4d+2)/\sqrt{Kp} &= o\left(\frac{(pn/\log n)^{3/4}/\log^{1/4}(pn)}{(pn \log n \log(pn))^{1/4}}\right)\\
&= o\left(\sqrt{\frac{pn}{\log^2n \log(pn)}}\right).
\end{align*}
By \autoref{lem:binomial_deviation}, $\alpha = \Omega(1)$.  With these two facts, we deduce that we can find $x$ such that $\alpha x / 2 \geq (4d+2)/\sqrt{Kp}$.

Now for this choice of $x$ we have 
\begin{align*}
\exp&\left(\frac{-(\alpha x - (4d+2)/\sqrt{Kp})^2}{2\alpha x}\right) \\
&\leq \exp\left(\frac{-(\alpha x - \alpha x / 2)^2}{2\alpha x}\right) \\
&= \exp\left(-\alpha x / 8\right) = \exp(-\omega(\log n)).
\end{align*}

Thus with probability at least $1 - \exp(-\omega(\log n))$, $\sum_{i=1}^x|U_i-V_i| \geq 4d+3$. This implies that the multiset of the $G'$-degrees of the vertices in $G'$ connected to $u$ whose $G'$-degrees are at most $pn-1$ has at least $4d+3$ elements different from the multiset of the $G'$-degrees of the vertices in $G'$ connected to $v$ whose $G'$-degrees are at most $pn-1$.  By \autoref{lem:indep_disjoint}, $u$ and $v$'s degree neighborhoods in $G$ are $(pn,4d+1)$-disjoint.  We then union bound over all $O(n^2)$ pairs of vertices to achieve the theorem.  
\end{proof}

\end{document}